\documentclass[letterpaper, 10 pt, conference]{ieeeconf}  

\IEEEoverridecommandlockouts                              
\usepackage{graphicx}

\usepackage{epsfig} 
\usepackage{times} 
\usepackage{amsmath} 
\usepackage{amsfonts}
\usepackage{mathtools}
\usepackage{amssymb}  

\usepackage{xcolor} 

\usepackage[caption=false,font=footnotesize]{subfig}

\usepackage{lipsum}

\usepackage{amsthm}

\newtheoremstyle{mythmstyle}
  {4pt}   
  {4pt}   
  {\normalfont}
  {}
  {\bfseries}
  {.}
  { }
  {}

\theoremstyle{mythmstyle}

\newtheorem{theorem}{Theorem}
\newtheorem{lemma}{Lemma}

\newtheorem{corollary}{Corollary}

\newtheorem{assumption}{Assumption}

\newcommand{\R}{\mathbb{R}}

\newcommand{\Ncal}{\mathcal{N}}

\newcommand{\Jcal}{\mathcal{J}}
\newcommand{\one}{\mathbf{1}}

\newcommand{\qt}{\mathsf{q}}
\newcommand{\pt}{\mathsf{p}}

\newcommand{\yt}{\mathsf{y}}

\newcommand{\xt}{\mathsf{x}}
\newcommand{\st}{\mathsf{s}}
\newcommand{\zt}{\mathsf{z}}
\newcommand{\wt}{\mathsf{w}}

\newcommand{\vt}{\mathsf{v}}

\title{\LARGE \bf
Emergent Behavior Is Robust to Communication Delays\\at the Cost of Slower System Evolution
}

\author{Seokho Jeong, Jin Gyu Lee, and Hyungbo Shim
\thanks{This work was supported by the National Research Foundation of Korea (NRF) grant funded by the Korea government (MSIT) (RS-2026-25504174).}
\thanks{Seokho Jeong, Jin Gyu Lee, and Hyungbo Shim are with ASRI, Department of Electrical and Computer Engineering, Seoul National University, Seoul 08826, South Korea (e-mail: {\tt \small shjeong@cdsl.kr}; {\tt \small jingyu.lee@snu.ac.kr}; {\tt \small  hshim@snu.ac.kr}).}
}
\begin{document}

\maketitle
\thispagestyle{empty}
\pagestyle{empty}

\begin{abstract}

Previous works have shown that strong coupling among heterogeneous agents enforces practical synchronization, leading to collective behavior governed by emergent dynamics.
However, it remains an open question whether emergent behavior persists in the presence of communication delays, as high-gain methods are typically sensitive to delays.
To address this problem, we instead slow down the agent dynamics, reproducing the synchronization mechanism of strong coupling on a slow time scale.
As a result, the emergence of collective behavior is guaranteed despite arbitrary constant communication delays.
This, however, comes at a cost: delays induce a derivative-like term, thereby scaling down the vector field of the emergent dynamics.
The amount of the scaling is explicitly characterized as a function of the coupling weights and the delay lengths.
We further suggest that this delay-induced slowdown, which is of independent interest, may be a general feature of diffusive coupling under suitable assumptions.

\end{abstract}

\section{Introduction}\label{section:introduction}

Motivated by observations in nature and engineering applications, researchers have studied the collective behavior of multi-agent systems for decades through the lenses of consensus and synchronization.
Following the seminal work \cite{olfati-saber2004tac}, the literature has expanded to encompass more general classes of agent dynamics \cite{scardovi2008cdc,yu2011tac,wieland2011aut} and communication constraints \cite{moreau2004cdc,munz2010aut,dimarogonas2012tac}. Within this landscape, heterogeneous networks are of particular interest because, unlike networks of identical agents, they can exhibit rich collective behavior distinct from that of any individual agent in isolation.
Characterizing this collective behavior is important since it often encodes global information that each agent cannot access directly, potentially leading to design principles for distributed algorithms.

For example, consider a network of $N$ agents modeled by
\begin{align}\label{eq:intro_high_gain}
    \dot{x}_i=f_i(t,x_i)+k\sum_{j\in \Ncal_i}a_{ij}(x_j-x_i),\quad i=1,\ldots,N,
\end{align}
where $x_i \in \R^n$ is the state of agent $i$, $k>0$ is the coupling gain, $\Ncal_i$ is the set of neighbors of agent $i$, and $a_{ij}$ is the $(i,j)$-th element of the adjacency matrix associated with the communication graph (see Section~\ref{subsection:problem_setup}). Each $f_i$ represents the node vector field of agent $i$, which can differ across agents.
Due to this heterogeneity, exact state synchronization is generically unattainable.

Nonetheless, it was shown in \cite{kim2016tac,lee2020aut} that strong coupling (i.e., large $k$) enforces practical synchronization among agents, and the resulting collective behavior is governed by the so-called \emph{blended dynamics}
\begin{align}\label{eq:intro_blended_dynamics}
    \dot{s}=\frac{1}{N}\sum_{i=1}^N f_i(t,s),
\end{align}
provided that the blended dynamics is contractive (in the sense of Assumption~\ref{ass:contractivity}) and the communication graph is undirected and connected.
The right-hand side of \eqref{eq:intro_blended_dynamics} generally differs from any individual $f_i$, which makes the blended dynamics \emph{emergent}; see also \cite{panteley2017tac} for the term \emph{emergent dynamics} and related results under strong coupling. Moreover, it aggregates each $f_i$ into a single vector field, synthesizing local information into global behavior.
This result has been leveraged in applications such as distributed optimization and distributed observer design, where heterogeneity is used as an intentional design tool; see~\cite{lee2022bookchapter} and the references therein.

Despite its importance, the theory of blended dynamics has a limitation: it relies on strong coupling, which can make the emergent behavior fragile under communication constraints. In particular, it has been unknown whether the emergent behavior persists under communication delays, because strong coupling may amplify the errors induced by delayed information. As delays are typically given by physical limits and cannot be adjusted, the use of strong coupling can be impractical when delays are non-negligible.

To address this issue, we propose an alternative model in which the node dynamics of each agent $i=1,\ldots,N$ is slowed down:
\begin{align}\label{eq:system}
\begin{split}
    \dot{x}_i(t) = \varepsilon f_i(\varepsilon t,x_i(t)) + \sum_{j\in \Ncal_i} a_{ij}(x_j(t-d_{ij})-x_i(t)),
\end{split}
\end{align}
where $0<\varepsilon \le 1$ is the parameter that controls the overall evolution speed of agents, and each $d_{ij} \ge 0$ denotes the communication delay from agent~$j$ to agent~$i$.
Analysis on the slow time scale $\tau \coloneqq \varepsilon t$ reveals that slow evolution reduces the delay-induced errors, allowing diffusive coupling perturbed by delays to still enforce practical synchronization.
We show that if $\varepsilon$ is sufficiently small, collective behavior emerges even under large constant communication delays. Moreover, we prove that delays modify the resulting blended dynamics by scaling down its vector field, yielding
\begin{align}\label{eq:blended_dynamics}
    \st' = \frac{1}{1 + \frac{1}{N}\sum_{i,j=1}^N a_{ij}d_{ij}} \cdot\frac{1}{N} \sum_{i=1}^N f_i(\tau,\st),
\end{align}
where the prime denotes the $\tau$-derivative, i.e., $\st'=d\st/d\tau$.
This formula is somewhat surprising, since a naive analogue of \eqref{eq:intro_blended_dynamics} would be $\dot{s} = (1/N) \sum_{i=1}^N \varepsilon f_i(\varepsilon t, s)$, which yields $\st' = (1/N) \sum_{i=1}^N f_i(\tau, \st)$ on the $\tau$-scale.
We also relate our results to the strong-coupling model, showing that emergent behavior induced by strong coupling persists under sufficiently small delays, with the same scaling effect as in \eqref{eq:blended_dynamics}.

\textit{Notation:} For vectors $e_1,\ldots,e_m$, $\mathrm{col}(e_1,\ldots,e_m) \coloneq [e_1^T,\ldots,e_m^T]^T$. Let $\one_N \in \R^N$ denote the vector with all entries equal to one. We use $\|\cdot\|$ for the Euclidean norm of vectors and the induced 2-norm of matrices, and $\|\cdot\|_\infty$ for the infinity norm. For a symmetric positive definite matrix $H$, $\lambda_{\min}(H)$ denotes its smallest eigenvalue and $\|x\|_H\coloneqq\sqrt{x^THx}$. The space of continuous functions from $X$ to $Y$ is denoted by $C(X,Y)$. For a function $x$ defined on $\R$, $x[a,b]$ denotes the restriction of $x$ to $[a,b]$. The upper right-hand Dini derivative operator is denoted by $D^+$.

\section{Intuition on the Modified Blended Dynamics}\label{section:problem_setup}

\subsection{Problem Setup}\label{subsection:problem_setup}

We consider a network of $N$ heterogeneous agents with constant communication delays described by \eqref{eq:system}. Each neighborhood set $\Ncal_i$ defines the communication graph where agent $j$ transmits information to agent $i$ if and only if $j \in \Ncal_i$ (we adopt the convention $i \notin \Ncal_i$). The adjacency matrix $[a_{ij}]\in \R^{N\times N}$ imposes weights on the communication graph and satisfies $a_{ij}>0$ if $j \in \Ncal_i$ and $a_{ij}=0$ otherwise. Delays may be heterogeneous and asymmetric across communication links, and satisfy $0\le d_{ij} \le d$ where $d>0$ can be arbitrarily large. The initial histories satisfy $x_{i}[-d,0] \in C([-d,0],\R^n)$ for $i=1,\ldots,N$.

The goal of this paper is to show that the agents achieve practical synchronization for $0<\varepsilon \ll 1$, and to characterize the modified blended dynamics that governs the resulting collective behavior. First, we state some assumptions on the system's regularity and communication topology.

\begin{assumption}\label{ass:vector_field}
    Each $f_i(t,x_i)$ is continuous in $t$, locally Lipschitz in $x_i$ uniformly in $t$, and $C^1$ in $x_i$. Each $f_i(t,0)$ is uniformly bounded in $t$.
\end{assumption}

\begin{assumption}\label{ass:Lipschitz_continuity}
    For any compact set $K \subseteq \R^n$, there exists a corresponding $L>0$ such that
    \begin{align*}
        \|f_i(t_1,x_i)-f_i(t_2,x_i)\|_\infty \le L|t_1-t_2|, \qquad \forall x_i \in K
    \end{align*}
    for all $t_1,t_2 \in [0,\infty)$ and $i=1,\ldots,N$.
\end{assumption}

Assumption~\ref{ass:vector_field} imposes a mild regularity on each $f_i$. Assumption~\ref{ass:Lipschitz_continuity} additionally requires each $f_i$ to be Lipschitz continuous in $t$ uniformly in $x_i$, on each compact subset of $\R^n$. This ensures that the derivative $\st'(\tau)$ of a solution of \eqref{eq:blended_dynamics} is uniformly continuous when the solution is bounded.

\begin{assumption}\label{ass:graph}
    The communication graph associated with $[a_{ij}] \in \R^{N\times N}$ is strongly connected and balanced.
\end{assumption}

The balanced graph assumption in Assumption~\ref{ass:graph} is imposed only to simplify the expression of the blended dynamics. For general directed graphs containing a spanning tree, the blended dynamics takes the form of a weighted sum of the $f_i$, where the weights are determined by a left eigenvector of the graph Laplacian matrix associated with the zero eigenvalue.

For further analysis, we introduce the slow time scale $\tau = \varepsilon t$. Henceforth, for a signal $p:\R \to \R^m$, we denote its \mbox{$\tau$-scale} representation by a sans-serif symbol $\pt$, defined by $\pt(\tau) \coloneqq p(\tau/\varepsilon)$.
On the $\tau$-scale, \eqref{eq:system} becomes
\begin{align}\label{eq:tau_scale}
\begin{split}
    \xt'_i(\tau) = f_i(\tau,\xt_i(\tau)) + \frac{1}{\varepsilon}\sum_{j \in \Ncal_i}a_{ij}(\xt_j(\tau-\varepsilon d_{ij}) - \xt_i(\tau)),\\
    \xt_{i}[-\varepsilon d,0] \in C([-\varepsilon d,0],\R^n), \qquad i=1,\ldots,N.
\end{split}
\end{align}
This time scale conversion reveals that slowing down the node dynamics on the $t$-scale translates to an increased coupling gain and reduced delays on the $\tau$-scale.

\subsection{Modified Blended Dynamics}

To motivate the modified blended dynamics \eqref{eq:blended_dynamics}, we introduce the following coordinate transformation for \eqref{eq:tau_scale}:
\begin{align}\label{eq:coordinate_change}
    \begin{bmatrix}
        \zt\\
        \wt
    \end{bmatrix}
    =\left(
    \begin{bmatrix}
        \frac{1}{N}\one_N^T\\
        R^T
    \end{bmatrix}\otimes I_n \right)
    \xt, \quad \xt = \one_N \otimes \zt+(R\otimes I_n)\wt,
\end{align}
where $\xt\coloneqq\mathrm{col}(\xt_1,\ldots,\xt_N)$ and $R\in \R^{N\times(N-1)}$ satisfies $R^T\one_N=0$ and $R^TR=I_{N-1}$. Note that $\wt$ measures the disagreement among agents, in the sense that $\wt=0$ if and only if $\xt_1=\cdots=\xt_N$. When $\wt=0$, $\zt$ describes the synchronized behavior of the agents.
Let $R_i$ denote the \mbox{$i$-th} row of $R$. Using \eqref{eq:tau_scale} and Assumption~\ref{ass:graph}, the $\zt$-dynamics is given by
\begin{align}
    \begin{split}\label{eq:z_dynamics}
        \zt'(\tau) &= \frac{1}{N}\sum_{i=1}^N f_i(\tau,\zt(\tau)+(R_i\otimes I_n)\wt(\tau))\\
        & + \frac{1}{N}\sum_{i,j=1}^Na_{ij}\frac{\zt(\tau-\varepsilon d_{ij})-\zt(\tau)}{\varepsilon}\\
        &+\frac{1}{N}\sum_{i,j=1}^Na_{ij}(R_j \otimes I_n)\frac{\wt(\tau -\varepsilon d_{ij})-\wt(\tau)}{\varepsilon},
    \end{split}
\end{align}
where we used $\sum_{i,j=1}^Na_{ij}R_j=\sum_{i,j=1}^Na_{ij}R_i$. To focus on the case when all the agents are synchronized, we put $\wt \equiv 0$ and $\zt=\hat{\zt}$ in \eqref{eq:z_dynamics} to get the reduced model given by
\begin{align}\label{eq:z_tilde_dynamics}
\begin{split}
    \hat{\zt}'(\tau) &= \frac{1}{N}\sum_{i=1}^N f_i(\tau,\hat{\zt}(\tau))+ \frac{1}{N}\!\sum_{i,j=1}^Na_{ij}\frac{\hat{\zt}(\tau-\varepsilon d_{ij})-\hat{\zt}(\tau)}{\varepsilon}.
\end{split}
\end{align}

However, \eqref{eq:z_tilde_dynamics} is an intermediate model that still depends on $\varepsilon$. The key observation is that, for small $\varepsilon$, the last term in \eqref{eq:z_tilde_dynamics} behaves like a derivative of $\hat{\zt}$ multiplied by a constant, which follows from the approximation
\begin{align}\label{eq:approximation}
    \frac{\hat{\zt}(\tau-\varepsilon d_{ij})-\hat{\zt}(\tau)}{\varepsilon} \approx -d_{ij}\hat{\zt}'(\tau).
\end{align}
This motivates us to define the blended dynamics of \eqref{eq:system} or \eqref{eq:tau_scale} by replacing the $\hat{\zt}(\tau-\varepsilon d_{ij})$ term in \eqref{eq:z_tilde_dynamics} with its first-order approximation at $\tau$, which leads to \eqref{eq:blended_dynamics}. Although this is a heuristic derivation, we rigorously justify in Section~\ref{section:proof_of_theorem_1} that \eqref{eq:blended_dynamics} is indeed the governing equation of the synchronized behavior of \eqref{eq:tau_scale} under suitable assumptions.

\section{Main Result}\label{section:main_result}
Let $\bar{f}(\tau,\st)\coloneq \frac{1}{N}\sum_{i=1}^N f_i(\tau,\st)$.
Then \eqref{eq:blended_dynamics} can be written as $\st'=\frac{1}{1+A}\bar{f}(\tau,\st)$, where $A\coloneq \frac{1}{N}\sum_{i,j=1}^N a_{ij}d_{ij}$.

\begin{assumption}\label{ass:contractivity}
    There exist $\gamma>0$ and a symmetric positive definite matrix $H \in \R^{n\times n}$ such that
    \begin{align*}
        H \frac{\partial \bar{f}}{\partial \st}(\tau,\st) + \frac{\partial \bar{f}}{\partial \st}^T\!(\tau,\st)H \preceq -\gamma H
    \end{align*}
    for all $\tau \in [0,\infty)$ and $\st \in \R^n$.
\end{assumption}

Assumption~\ref{ass:contractivity} imposes the contraction property on \eqref{eq:blended_dynamics} so that any two solutions of \eqref{eq:blended_dynamics} with different initial conditions approach each other at an exponential rate \cite{lohmiller1998aut}. It allows \eqref{eq:blended_dynamics} to describe the agent trajectories on the infinite time interval, even under heterogeneity. Combined with Assumption~\ref{ass:vector_field}, it yields other useful properties, which are summarized in the following lemma.

\begin{lemma}\label{lem:s_dynamics}
    Suppose Assumptions~\ref{ass:vector_field} and \ref{ass:contractivity} hold. Let
    \begin{align*}
        S\coloneqq \bigg\{\st \in \R^n:\|\st\|_H \le \frac{2(1+\bar{a}d)}{\gamma}\sup_{\tau \ge 0}\|\bar{f}(\tau,0)\|_H \bigg\},
    \end{align*}
    where $\bar{a}\coloneqq \frac{1}{N}\sum_{i,j=1}^N a_{ij}$. Then for any $\st(0) \in S$ and $d_{ij} \in [0,d]$, the solution $\st(\tau)$ of \eqref{eq:blended_dynamics} remains in $S$ and is Lipschitz continuous on $[0,\infty)$ with a uniform Lipschitz constant.
\end{lemma}

\begin{proof}
    Assumption~\ref{ass:vector_field} guarantees that $\sup_{\tau \ge 0}\|\bar{f}(\tau,0)\|_H$ is finite. Let $\st(\tau)$ be a solution of \eqref{eq:blended_dynamics} with $\st(0) \in S$. For any $d_{ij} \in [0,d]$, we have
    \begin{align*}
        \frac{d}{d\tau}\|\st\|^2_H &= \frac{2}{1+A}\left[ \st^TH (\bar{f}(\tau,\st)-\bar{f}(\tau,0))+\st^TH\bar{f}(\tau,0)\right]\\
        &\le  \frac{1}{1+A}\left( -\gamma \|\st\|^2_H + 2\|\st\|_H\sup_{\tau \ge 0}\|\bar{f}(\tau,0)\|_H \right)\\
        & \le \|\st\|_H\left(-\frac{\gamma}{1+\bar{a}d}\|\st\|_H + 2\sup_{\tau \ge 0}\|\bar{f}(\tau,0)\|_H \right),
    \end{align*}
    where the first inequality follows from Assumption~\ref{ass:contractivity}. This shows $\st(\tau) \in S$ for all $\tau \ge 0$. Since $S$ is bounded, the derivative $\st'(\tau)\!=\!\frac{1} {1+A}\bar{f}(\tau,\st(\tau))$ is uniformly bounded for all $\tau \ge 0$ and $d_{ij} \in [0,d]$. Hence, $\st(\tau)$ is Lipschitz continuous for all $d_{ij} \in [0,d]$ with a uniform Lipschitz constant.
\end{proof}

We now state the main result of this paper.

\begin{theorem}\label{thm:main_result}
    Suppose Assumptions~\ref{ass:vector_field}--\ref{ass:contractivity} hold. Let $K \subseteq \R^n$ be compact and let $d>0$. Then there exist $\beta_0 > 0$ and $\varepsilon^* \in (0,1]$ such that the following holds for $0<\varepsilon\le\varepsilon^*$: for any constant delays $d_{ij}\in [0,d]$, every solution of \eqref{eq:system} with $x_i[-d,0] \in C([-d,0],K)$ for all $i=1,\ldots,N$ satisfies
    \begin{align*}
        \limsup_{t \to \infty} \|x_i(t)-\st(\varepsilon t)\| \le \beta_0 \varepsilon, \qquad i=1,\ldots,N,
    \end{align*}
    where $\st(\tau)$ is any solution of \eqref{eq:blended_dynamics}.
\end{theorem}

The proof of Theorem~\ref{thm:main_result} is given in Section~\ref{section:proof_of_theorem_1}. Theorem~\ref{thm:main_result} states that even under large constant communication delays, sufficiently small $\varepsilon$ ensures practical synchronization of the agents, with the synchronized trajectories approximately following the slow signal $\st(\varepsilon t)$, not $\st(t)$.
Hence, emergent behavior persists in the presence of delays when each agent's node dynamics is sufficiently slow to induce practical synchronization.
This result also provides an alternative demonstration of the idea that ``enforcing synchronization against heterogeneity leads to emergent behavior," previously illustrated in the context of strong coupling \cite{lee2022bookchapter}.

We emphasize that the collective behavior is governed by the modified blended dynamics \eqref{eq:blended_dynamics}, whose vector field is a scaled-down version of that of \eqref{eq:intro_blended_dynamics}.
This can be interpreted as delays making the collective behavior slower than in the delay-free case. As shown in \eqref{eq:approximation}, this phenomenon arises from the structure of diffusive coupling rather than from agent heterogeneity. In this regard, the slowdown of collective behavior caused by communication delays is a general feature of diffusive coupling, which may also be observed in homogeneous agents under suitable assumptions. This provides useful intuition for analyzing the collective behavior of multi-agent systems under delays.

For the scaling factor $(1+A)^{-1}$ in \eqref{eq:blended_dynamics} to be close to one, $A$ must be small. This shows the trade-off between coupling weights and delay lengths, as larger $d_{ij}$ requires smaller corresponding $a_{ij}$ to keep $A$ small. Furthermore, it highlights the advantage of sparse graphs where each agent has a relatively small number of neighbors, which reduces the number of summands contributing to $A$.

If each $f_i$ is autonomous, i.e., $f_i=f_i(x_i)$, then every solution of \eqref{eq:blended_dynamics} is a time-rescaled version of a solution of \eqref{eq:intro_blended_dynamics}.
This explains the delay robustness of distributed algorithms based on the autonomous blended dynamics. For example, in distributed optimization, one often designs the delay-free blended dynamics \eqref{eq:intro_blended_dynamics} to have a globally asymptotically stable equilibrium point corresponding to a minimizer of an objective function \cite{lee2022aut}. Although delays alter the blended dynamics from \eqref{eq:intro_blended_dynamics} to \eqref{eq:blended_dynamics}, its solutions converge to the same minimizer in this case.

Finally, we relate Theorem~\ref{thm:main_result} to the strong-coupling model. Consider a network of $N$ agents ($i=1,\ldots,N$):
\begin{align}\label{eq:high_gain_with_delay}
    \dot{x}_i(t)=f_i(t,x_i(t))+k\sum_{j\in \Ncal_i}a_{ij}(x_j(t-\tilde{d}_{ij})-x_i(t)),
\end{align}
where we interpret $\tilde{d}_{ij}=d_{ij}/k$ to apply Theorem~\ref{thm:main_result} with $\varepsilon = 1/k$; compare \eqref{eq:high_gain_with_delay} with \eqref{eq:tau_scale}. The corresponding blended dynamics is given by
\begin{align}\label{eq:high_gain_blended_dynamics}
    \dot{s}=\frac{1}{1+\frac{1}{N}\sum_{i,j=1}^N ka_{ij}\tilde{d}_{ij}}\cdot\frac{1}{N}\sum_{i=1}^N f_i(t,s).
\end{align}
For \eqref{eq:high_gain_with_delay}, Theorem~\ref{thm:main_result} gives the following corollary, illustrating that emergent behavior is robust to sufficiently small delays.
\begin{corollary}\label{cor:high_gain}
    Suppose Assumptions~\ref{ass:vector_field}--\ref{ass:contractivity} hold. Let $K \subseteq \R^n$ be compact and let $d>0$. Then there exist $\beta_0 > 0$ and $k^*>0$ such that the following holds for $k\ge k^*$: for any constant delays $\tilde{d}_{ij} \in [0,d/k]$, every solution of \eqref{eq:high_gain_with_delay} with $x_i[-d/k,0] \in C([-d/k,0],K)$ for all $i=1,\ldots,N$ satisfies 
    \begin{align*}
        \limsup_{t \to \infty} \|x_i(t)-s(t)\| \le \beta_0/k, \qquad i=1,\ldots,N,
    \end{align*}
    where $s(t)$ is any solution of \eqref{eq:high_gain_blended_dynamics}.
\end{corollary}

\section{Simulations}\label{section:simulations}

\begin{figure}[t]
\centering

\subfloat[\label{subfig:large_k_small_delay}]{%
\includegraphics[width=0.49\linewidth]{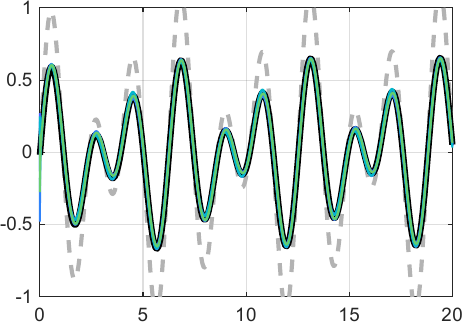}}
\hfill
\subfloat[\label{subfig:large_k_large_delay}]{%
\includegraphics[width=0.49\linewidth]{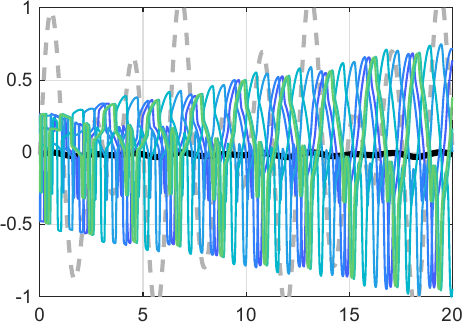}}

\caption{Trajectories of six agents governed by \eqref{eq:high_gain_with_delay} with coupling gain $k=100$. The black solid line shows the solution of \eqref{eq:high_gain_blended_dynamics}, and the gray dashed line shows the solution of \eqref{eq:intro_blended_dynamics}. Delays are i.i.d. in each case. (a)~$\tilde{d}_{ij}\sim\mathrm{Uniform}(0,0.01)$. (b)~$\tilde{d}_{ij}\sim\mathrm{Uniform}(0,1)$.}
\label{fig:strong_coupling}

\end{figure}

\begin{figure}[t]
\centering

\subfloat[\label{subfig:small_e_large_delay}]{%
\includegraphics[width=0.49\linewidth]{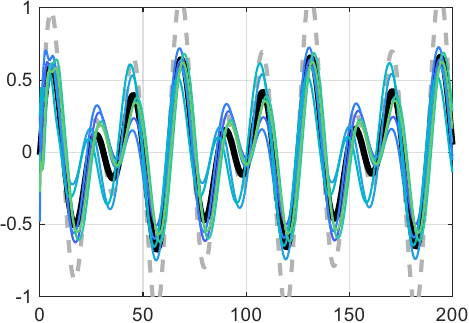}}
\hfill
\subfloat[\label{subfig:tiny_e_large_delay}]{%
\includegraphics[width=0.49\linewidth]{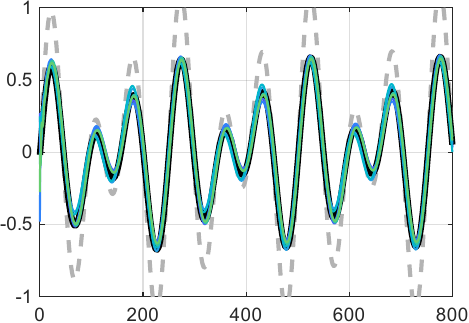}}

\caption{Trajectories of six agents governed by \eqref{eq:system} under the same delays as in Fig.~\ref{fig:strong_coupling}(b). The black solid line shows $\st(\varepsilon t)$ where $\st(\tau)$ is the solution of \eqref{eq:blended_dynamics}, and the gray dashed line shows $s(\varepsilon t)$ where $s(t)$ is the solution of \eqref{eq:intro_blended_dynamics}. Note the scaling of the time axis. (a)~$\varepsilon=0.1$. (b)~$\varepsilon=0.025$.}
\label{fig:slow_evolution}

\end{figure}

Consider a network of six scalar agents with
\begin{align*}
    f_1(t,x_1)&=-x_1-2\sin{t},&f_2(t,x_2)&=-2x_2+8\cos{2t},\\
    f_3(t,x_3)&=x_3+2\sin{t},&f_4(t,x_4)&=-3x_4+12\cos{3t},\\
    f_5(t,x_5)&=2x_5,&f_6(t,x_6)&=x_6-2\cos{2t}.
\end{align*}
The communication graph is a directed cycle with information flow $1 \leftarrow 2,2 \leftarrow 3,\ldots,6\leftarrow 1$ and weights $a_{ij}=2$ if $j \in \Ncal_i$ and $a_{ij}=0$ otherwise. Each agent is initialized with a constant history sampled uniformly from $(-0.5,0.5)$.
Observe that $\bar{f}(t,s)=-\frac{1}{3}s + \cos{2t}+2\cos{3t}$ satisfies Assumption~\ref{ass:contractivity}.
We initialize \eqref{eq:intro_blended_dynamics}, \eqref{eq:blended_dynamics} and \eqref{eq:high_gain_blended_dynamics} at the average of the initial states $x_i(0)$ as a baseline choice.

We first examine the effect of communication delays in the strong-coupling model \eqref{eq:high_gain_with_delay}. Fig.~\ref{fig:strong_coupling} shows the trajectories of six agents governed by \eqref{eq:high_gain_with_delay} with $k=100$ and different delay ranges; the delays $\tilde{d}_{ij}$ are sampled i.i.d. from $\mathrm{Uniform}(0,0.01)$ in case (a) and $\mathrm{Uniform}(0,1)$ in case (b). In case (a), the solution of the modified blended dynamics \eqref{eq:high_gain_blended_dynamics} (black solid line) accurately captures the synchronized motion of the agents. In contrast, the agents fail to synchronize in case (b) due to the non-negligible delays.

Next, we simulate the slow-evolution model \eqref{eq:system} using the same sampled delays as in Fig.~\ref{fig:strong_coupling}(b).
Fig.~\ref{fig:slow_evolution} shows the trajectories of six agents governed by \eqref{eq:system} with $\varepsilon=0.1$ in case (a) and $\varepsilon=0.025$ in case (b). Although the delays are non-negligible, decreasing $\varepsilon$ leads to practical synchronization, where the synchronized behavior is governed by the modified blended dynamics \eqref{eq:blended_dynamics}.

Finally, in all cases shown in Fig.~\ref{fig:strong_coupling} and \ref{fig:slow_evolution}, the solution of the delay-free blended dynamics \eqref{eq:intro_blended_dynamics} (gray dashed line) fails to capture the collective behavior of the agents.

\section{Proof of Theorem~\ref{thm:main_result}}\label{section:proof_of_theorem_1}

Although the proof of Theorem~\ref{thm:main_result} constitutes a substantial part of our contribution, we present only the key ideas and the main steps due to space limitations. Complete proofs of the following lemmas and Theorem~\ref{thm:main_result} can be found in the Appendix. Throughout this section, the symbols $\xt$, $\zt$, and $\wt$ are related by \eqref{eq:coordinate_change} and all Dini derivatives are taken with respect to $\tau$. For a trajectory $\pt:\R \to \R^m$ on the $\tau$-scale, we define the history segment $\pt_\tau \in C([-\varepsilon d,0],\R^m)$ by $\pt_\tau(\theta)\coloneqq \pt(\tau+\theta)$ for $\theta \in [-\varepsilon d,0]$.

\subsection{Useful Lemmas}
We first introduce lemmas that are frequently used in the analysis. For $\yt \coloneqq \mathrm{col}(\yt_1,\ldots,\yt_N)\in C([-\varepsilon d,0],\R^{Nn})$ where $\yt_i \coloneq \mathrm{col}(\yt_i^1,\ldots,\yt_i^n)$, $i=1,\ldots,N$, define
\begin{align*}
    M^{\ell}(\yt) \coloneq \max_{\substack{i=1,\ldots,N \\ \theta \in [-\varepsilon d,0]}} \yt_i^{\ell}(\theta),\quad m^{\ell}(\yt) \coloneq \min_{\substack{i=1,\ldots,N \\ \theta \in [-\varepsilon d,0]}}\yt_i^{\ell}(\theta),
\end{align*}
for $\ell=1,\ldots,n$. Also define
\begin{align*}
    E(\yt) \coloneq \max_{\ell = 1,\ldots,n}\left\{ M^{\ell}(\yt) - m^{\ell}(\yt) \right\}.
\end{align*}
The functional $E$ measures disagreement among agent histories, playing a role similar to that of $\wt$. The following lemma summarizes useful properties of $E$.
\begin{lemma}\label{lem:property_of_E}
    The functional $E$ is a seminorm such that $E(\yt) \allowbreak \le 2\sup_{\theta \in [-\varepsilon d,0]}\|\yt(\theta)\|_\infty$ for all $\yt \in C([-\varepsilon d,0],\R^{Nn})$. Also, $\|\wt(\tau)\| \le \sqrt{Nn} E(\xt_\tau)$ for any $\xt:\R \to \R^{Nn}$.
\end{lemma}

Let $f\coloneqq \mathrm{col}(f_1,\ldots,f_N)$. The next lemma provides a rough $\varepsilon$-independent estimate for the evolution of \eqref{eq:tau_scale}.

\begin{lemma}\label{lem:evolution_bound}
    Let $\xt(\tau)$ be a solution of \eqref{eq:tau_scale}. Suppose we have $\|f(\tau,\xt(\tau))\|_\infty \le B_0$ for $\tau \in [\tau_0,\tau_1)$ where $0\le\tau_0<\tau_1$. Then for $\ell=1,\ldots,n$,
    \begin{align*}
        M^{\ell}(\xt_{\tau_1}) &\le B_0(\tau_1-\tau_0)  + M^{\ell}(\xt_{\tau_0})\\
        m^{\ell}(\xt_{\tau_1}) &\ge -B_0(\tau_1-\tau_0)  + m^{\ell}(\xt_{\tau_0}).
    \end{align*}
\end{lemma}

Next, we provide a Lyapunov-Krasovskii functional for estimating synchronization error.
The construction is similar to that in \cite[Lemma 33.1]{krasovskii} and relies on the fact that consensus of integrators is robust to communication delays~\cite{munz2011tac}.

\begin{lemma}\label{lem:output_converse_Lyapunov_Krasovskii}
    Suppose Assumption~\ref{ass:graph} holds. Consider a network of $N$ integrators described by
    \begin{align}\label{eq:integrator_tau}
    \begin{split}
        \vt'_i(\tau) = \frac{1}{\varepsilon}\sum_{j\in \Ncal_i} a_{ij}(\vt_j(\tau-\varepsilon d_{ij})-\vt_i(\tau)) \quad \in \R^n,\\
        \vt_{i}[-\varepsilon d,0] \in C([-\varepsilon d,0],\R^n), \qquad i=1,\ldots,N.
    \end{split}
    \end{align}
    Then there exist a functional $V:C([-\varepsilon d,0],\R^{Nn}) \to [0,\infty)$ and constants $c_2,c_4>0$ such that for every $\yt,\tilde{\yt} \in C([-\varepsilon d,0],\R^{Nn})$ and for any solution $\vt(\tau)$ of \eqref{eq:integrator_tau},
    \begin{align*}
        &E(\yt) \le V(\yt) \le c_2E(\yt)\\
        &D^+V(\vt_\tau) \le -\frac{1}{2\varepsilon} E(\vt_\tau)\\
        &|V(\yt)-V(\tilde{\yt})| \le c_4E(\yt-\tilde{\yt})
    \end{align*}
    for $\tau \ge 0$. Moreover, constants $c_2,c_4$ do not depend on $\varepsilon$ and the particular values of $d_{ij} \in [0,d]$.
\end{lemma}

Using Lemma~\ref{lem:output_converse_Lyapunov_Krasovskii} and techniques from \cite{burton,yeganefar2008tac}, we obtain the following estimate for the synchronization error.

\begin{lemma}\label{lem:practical_consensus}
    Let $V$ be as in Lemma~\ref{lem:output_converse_Lyapunov_Krasovskii} and let $\xt(\tau)$ be a solution of \eqref{eq:tau_scale}. Then for $\tau \ge 0$,
    \begin{align*}
        D^+V(\xt_\tau) & \le -\frac{1}{2\varepsilon c_2}V(\xt_\tau) + 2c_4 \|f(\tau,\xt(\tau))\|_\infty.
    \end{align*}
    In particular, if $\|f(\tau,\xt(\tau))\|_\infty \le B_0$ for $\tau \in [\tau_0,\tau_1)$ where $0\le\tau_0<\tau_1$, then for $\tau \in [\tau_0,\tau_1)$,
    \begin{align*}
        E(\xt_\tau) \le \exp\left(-\frac{\tau-\tau_0}{2\varepsilon c_2} \right) V(\xt_{\tau_0}) + 4\varepsilon c_2 c_4 B_0.
    \end{align*}
\end{lemma}

\subsection{Key Steps and Estimates}

We may consider a single collection $\{d_{ij}\}_{i,j=1}^N$ for simplicity; the following estimates hold uniformly for $d_{ij} \in [0,d]$.
Choose a compact rectangle $K_1 \subseteq \R^n$ and a compact, convex set $K_2 \subseteq \R^n$ such that $K \cup S \subsetneq \operatorname{int}K_1$ and $ K_1 \subsetneq  \operatorname{int} K_2$, where $S$ is defined in Lemma~\ref{lem:s_dynamics}. By Assumption~\ref{ass:vector_field}, $B_0 \coloneqq \sup\{\|f_i(\tau,\xt_i)\|_\infty:\xt_i \in K_2, i=1,\ldots,N,\tau \ge 0\}$ is finite. Then, using Lemma~\ref{lem:evolution_bound}, we can choose $T>0$ such that $\xt_i(\tau) \in K_1,i=1,\ldots,N$ for $\tau \in [0,T]$ regardless of the value of $\varepsilon$. Set $T_1 \coloneqq T/3$ and let $\overline{T}_{\varepsilon}>0$ be the first time $\tau$ at which any $\xt_i(\tau),i=1,\ldots,N$ leaves $K_2$. By construction, we have $T_1<T<\overline{T}_{\varepsilon}$ for any $\varepsilon>0$; we later claim that $\overline{T}_{\varepsilon} = \infty$ for sufficiently small $\varepsilon>0$.

Let $V$ be as in Lemma~\ref{lem:output_converse_Lyapunov_Krasovskii}, and let $V_0$ be an upper bound for $V(\xt[-\varepsilon d,0])$ over all initial histories satisfying $\xt_i[-\varepsilon d,0] \in C([-\varepsilon d,0],K)$, $i=1,\ldots,N$. Using Lemma~\ref{lem:practical_consensus}, we obtain
\begin{align*}
    E(\xt_\tau) \le \exp\left(-\frac{\tau}{2\varepsilon c_2} \right)V_0 + 4\varepsilon c_2 c_4 B_0, \qquad  \forall\tau \in [0,\overline{T}_{\varepsilon}).
\end{align*}
This shows that $E(\xt_\tau)=O(\varepsilon)$, and hence $\|\wt(\tau)\| = O(\varepsilon)$ by Lemma~\ref{lem:property_of_E}, as $\varepsilon \to 0+$ uniformly for $\tau \in [T_1,\overline{T}_\varepsilon)$.

We now compare $\zt$ and $\st$. By Assumption~\ref{ass:contractivity}, it suffices to consider a solution $\st(\tau)$ with $\st(0) \in S$. Define $F(\tau,\zt,\wt)\coloneq \frac{1}{N}\sum_{i=1}^N f_i(\tau,\zt+(R_i\otimes I_n)\wt)$. Comparing \eqref{eq:blended_dynamics} with \eqref{eq:z_dynamics} and adding and subtracting appropriate terms, we obtain
\begin{align}\label{eq:final_z_s}
\begin{split}
    &\zt'-\st'=\left[ F(\tau,\zt,0) -F(\tau,\st,0) \right]+[F(\tau,\zt,\wt)-F(\tau,\zt,0)]\\
    &\quad+\frac{1}{\varepsilon N}\sum_{i,j=1}^N a_{ij}\left[\zt(\tau-\varepsilon d_{ij})-\st(\tau-\varepsilon d_{ij})-(\zt-\st)\right]\\
    &\quad+\frac{1}{N}\sum_{\substack{i,j=1 \\ d_{ij} \neq 0}}^N a_{ij}d_{ij}\left[\frac{1}{1+A}F(\tau,\st,0)- \frac{\st(\tau-\varepsilon d_{ij})-\st}{-\varepsilon d_{ij}} \right]\\
    &\quad +\frac{1}{N} \sum_{i,j=1}^N a_{ij}(R_j \otimes I_n)\frac{\wt(\tau-\varepsilon d_{ij})-\wt}{\varepsilon},
\end{split}
\end{align}
where we omit the argument $\tau$, e.g., $\zt$ instead of $\zt(\tau)$, when the meaning is clear from the context.

Let $W_1(\tau)\coloneq \|\zt(\tau)-\st(\tau)\|^2_H$ and let $l>0$ be an upper bound for $\left\| \frac{\partial F}{\partial \wt}(\tau,\zt,\wt) \right\|$ over all $\tau \ge 0$, $\zt \in K_2$, and $\|\wt\| \le 1$. Note that Assumption~\ref{ass:contractivity} yields
\begin{align*}
    (\zt-\st)^TH\left[ F(\tau,\zt,0) -F(\tau,\st,0) \right] \le -\frac{\gamma}{2}\|\zt-\st\|^2_H.
\end{align*}
Moreover, for small $\varepsilon$, we have $\|\wt(\tau)\|=O(\varepsilon) \le 1$ for $\tau \in  [T_1,\overline{T}_{\varepsilon})$.
Using the observations above and the Cauchy-Schwarz inequality repeatedly, one can deduce from \eqref{eq:final_z_s} that for sufficiently small $\varepsilon>0$,
\begin{align}\label{eq:final_DW}
\begin{split}
    &D^+\sqrt{W_1}  \le -\frac{\gamma}{2}\sqrt{W_1} + l\sqrt{\|H\|}\,\|\wt\|\\
    &+\frac{1}{\varepsilon N}\sum_{i,j=1}^N a_{ij}\left( \sqrt{W_1(\tau-\varepsilon d_{ij})}-\sqrt{W_1} \right)\\
    &+\frac{1}{N}\sum_{\substack{i,j=1 \\ d_{ij} \neq 0}}^N a_{ij}d_{ij} \left\|\frac{1}{1+A}F(\tau,\st,0)- \frac{\st(\tau-\varepsilon d_{ij})-\st}{-\varepsilon d_{ij}} \right\|_H \\
    &+\frac{\sqrt{\|H\|}}{N}\sum_{i,j=1}^N a_{ij}\left\| \frac{\wt(\tau- \varepsilon d_{ij})-\wt}{\varepsilon}\right\|, \qquad  \forall \tau \in [T,\overline{T}_{\varepsilon}).
\end{split}
\end{align}
From now on, we focus on the last three terms in \eqref{eq:final_DW}.

\textbf{Step 1.} The last term in \eqref{eq:final_DW} may seem difficult to control because of the division by $\varepsilon$. To handle this term, we observe, for each $h \in [0,d]$, the dynamics of the difference quotient
\begin{align*}
    \qt^{h}_i(\tau) \coloneqq \frac{\xt_i(\tau-\varepsilon h)-\xt_i(\tau)}{\varepsilon}.
\end{align*}
In particular, we show that $\{\qt^h_i\}_{i=1}^N$ also achieves practical synchronization.
The key observation is that, by the linearity of the diffusive coupling, $\{\qt_i^{h}\}_{i=1}^N$ inherits the same coupling structure. For $\tau \ge T_1$ and sufficiently small $\varepsilon>0$, we use \eqref{eq:tau_scale} to obtain
\begin{align}\label{eq:final_quotient_dynamics}
\begin{aligned}
    (\qt_i^{h})'(\tau)
    &=\frac{1}{\varepsilon}[f_i(\tau-\varepsilon h,\xt_i(\tau-\varepsilon h))-f_i(\tau,\xt_i(\tau))]\\
    &\quad +\frac{1}{\varepsilon}\sum_{j \in \Ncal_i}a_{ij}(\qt_j^{h}(\tau-\varepsilon d_{ij})-\qt_i^{h}(\tau)).
\end{aligned}
\end{align}

Set $\qt^{h} \coloneqq \mathrm{col}(\qt_1^{h},\ldots,\qt_N^{h})$. Using the boundedness of each $\qt_i^h(\tau)$ for $\tau \in  [T_1,\overline{T}_{\varepsilon})$ (this follows from $E(\xt_{\tau})=O(\varepsilon)$ but details are omitted) and the Lipschitz continuity of each~$f_i$ (Assumptions~\ref{ass:vector_field} and \ref{ass:Lipschitz_continuity}), one can show that there exists $B_1>0$ such that for $\tau \in [T_1,\overline{T}_{\varepsilon})$ and sufficiently small $\varepsilon>0$,
\begin{align*}
    \frac{1}{\varepsilon}\|f_i(\tau-\varepsilon h,\xt_i(\tau-\varepsilon h))-f_i(\tau,\xt_i(\tau))\|_{\infty} \le B_1,
\end{align*}
for any $h \in [0,d]$. Then, by applying Lemma~\ref{lem:practical_consensus} to \eqref{eq:final_quotient_dynamics} and using an argument similar to that used for showing $E(\xt_{\tau})=O(\varepsilon)$, one can obtain $E(\qt_\tau^h) = O(\varepsilon)$ as $\varepsilon \to 0+$ uniformly for $\tau \in  [T,\overline{T}_\varepsilon)$ and $h \in [0,d]$. By Lemma~\ref{lem:property_of_E}, it follows that
\begin{align*}
    \left\| \frac{\wt(\tau- \varepsilon d_{ij})-\wt}{\varepsilon}\right\| \le \sqrt{Nn}E(\qt_\tau^{d_{ij}}) = O(\varepsilon)
\end{align*}
as $\varepsilon \to 0+$ uniformly for $\tau \in  [T,\overline{T}_{\varepsilon})$ and $i,j=1,\ldots,N$.

\textbf{Step 2.} We now focus on the second-to-last term in \eqref{eq:final_DW}. When $d_{ij} \neq 0$ and $\tau -\varepsilon d > 0$, from \eqref{eq:blended_dynamics}, we have
\begin{align*}
    \frac{\st(\tau-\varepsilon d_{ij})-\st}{-\varepsilon d_{ij}}= \frac{1}{\varepsilon d_{ij}}\int_{\tau-\varepsilon d_{ij}}^{\tau} \frac{1}{1+A}F(\theta,\st(\theta),0)d\theta.
\end{align*}
By the Lipschitz continuity of each $f_i$ and the trajectory $\st(\tau)$ (Assumptions~\ref{ass:vector_field}--\ref{ass:Lipschitz_continuity} and Lemma~\ref{lem:s_dynamics}), we obtain that for each $(i,j)$ pair satisfying $d_{ij} \neq 0$,
\begin{align*}
\begin{split}
    &\bigg\|\frac{1}{1+A}F(\tau,\st(\tau),0)- \frac{\st(\tau-\varepsilon d_{ij})-\st}{-\varepsilon d_{ij}} \bigg\|_H\\
    &= \frac{1}{\varepsilon d_{ij}} \bigg\| \frac{1}{1+A} \int_{\tau -\varepsilon d_{ij}}^{\tau} (F(\tau,\st(\tau),0)-F(\theta,\st(\theta),0)) d\theta\bigg\|_H\\
    &\le \frac{1}{1+A}\sup_{\theta \in [\tau-\varepsilon d,\tau]} \! \left\|F(\tau,\st(\tau),0)-F(\theta,\st(\theta),0)\right\|_H \!=O(\varepsilon)
\end{split}
\end{align*}
as $\varepsilon \to 0+$ uniformly for $\tau \in  [T,\overline{T}_{\varepsilon})$.

\textbf{Step 3.} Finally, to cancel the remaining delayed term in \eqref{eq:final_DW}, we define
\begin{align*}
    W_2(\tau) \coloneqq \frac{1}{\varepsilon N} \sum_{i,j=1}^N a_{ij} \int_{\tau-\varepsilon d_{ij}}^\tau e^{-c(\tau-\varepsilon d_{ij}-\theta)}\sqrt{W_1(\theta)}\ d\theta,
\end{align*}
where $c>0$ will be chosen later. Let $W\coloneq\sqrt{W_1} +W_2$ and $\bar{a} \coloneqq \frac{1}{N}\sum_{i,j=1}^N a_{ij}$. Combining \eqref{eq:final_DW} with the results of the previous steps, we obtain for sufficiently small $\varepsilon>0$ that
\begin{align*}
\begin{split}
    D^+W &= D^+\sqrt{W_1} -cW_2\\
    &+ \frac{1}{\varepsilon N}\sum_{i,j=1}^N a_{ij}\bigg(e^{\varepsilon c d_{ij}}\sqrt{W_1}-\sqrt{W_1(\tau-\varepsilon d_{ij})}\bigg)\\
    &\le \bigg(-\frac{\gamma}{2} + \bar{a}\frac{e^{\varepsilon c d}-1}{\varepsilon} \bigg)\sqrt{W_1} - cW_2 + O(\varepsilon)
\end{split}
\end{align*}
uniformly for $\tau \in  [T,\overline{T}_\varepsilon)$. Since $\lim_{\varepsilon \to 0+} \frac{e^{\varepsilon c d}-1}{\varepsilon} = cd$, we can choose an appropriate value of $c$ so that the coefficient of $\sqrt{W_1}$ in the above inequality becomes smaller than $-c$ as $\varepsilon \to 0+$. For sufficiently small $\varepsilon>0$, this leads to
\begin{align}\label{eq:final_DW_simplified}
\begin{split}
    D^+W(\tau) 
    & \le -cW(\tau) + O(\varepsilon), \qquad \forall \tau \in [T,\overline{T}_{\varepsilon}).
\end{split}
\end{align}

By carefully choosing $K_2$ at the beginning of the proof and using \eqref{eq:final_DW_simplified}, one can show that $\overline{T}_{\varepsilon}=\infty$ for all sufficiently small $\varepsilon>0$. See Appendix~\ref{subsection:approximation_by_blended_dynamics} for the details of this argument. With this fact, the comparison principle yields
\begin{align*}
\begin{split}
    &\limsup_{\tau \to \infty} \|\xt_i(\tau)-\st(\tau)\|\\
    &\le \limsup_{\tau \to \infty} \|\zt(\tau)-\st(\tau)\| + \limsup_{\tau \to \infty}\|(R_i\otimes I_n)\wt(\tau)\| \\
    &\le \limsup_{\tau \to \infty}\frac{W(\tau)}{\sqrt{\lambda_{\min}(H)}} + \limsup_{\tau \to \infty}\|\wt(\tau)\|=O(\varepsilon)
\end{split}
\end{align*}
for sufficiently small $\varepsilon>0$. This completes the proof.

\section{Conclusion}

This paper shows that emergent behavior of heterogeneous agents persists in the presence of delays when each agent's node dynamics is sufficiently slow.
Moreover, delays additionally slow down the resulting collective behavior by scaling down the vector field of the blended dynamics.

Future work includes extending the analysis to more general settings, such as time-varying delays and different stability assumptions on the blended dynamics.
Relatedly, it has been observed that delays decrease the collective frequency in oscillator networks \cite{niebur1991PhysRevLett,torres2026PhysRevE}.
It would also be interesting to analyze this phenomenon as a continuation of our work, particularly when the blended dynamics has a stable limit cycle as in \cite{lee2018cdc}.

\bibliographystyle{IEEEtran}
\bibliography{Ref}

@book{burton,
  author = {T.~A.~Burton},
  year = {1985},
  title = {Stability and Periodic Solutions of Ordinary and Functional Differential Equations},
  publisher = {Academic Press},
  address = {Orlando, FL, USA}
}

@book{krasovskii,
  author = {N.~N.~Krasovskii},
  year = {1963},
  title = {Stability of Motion},
  publisher = {Stanford University Press},
  address = {trans. J.~L.~Brenner. Stanford, CA, USA},
  translator = {J.~L.~Brenner}
}

@article{yeganefar2008tac,
  author = {N.~Yeganefar and P.~Pepe and M.~Dambrine},
  title = {Input-to-state stability of time-delay systems: {A} link with exponential stability},
    year = {2008},
    volume = {53},
    number = {6},
    journal = {IEEE Trans. Autom. Control},
    pages = {1526--1531},
}

@article{munz2011tac,
  author = {U.~M\"{u}nz and A.~Papachristodoulou and F.~Allg\"{o}wer},
  title = {Consensus in multi-agent systems with coupling delays and switching topology},
    year = {2011},
    volume = {56},
    number = {12},
    journal = {IEEE Trans. Autom. Control},
    pages = {2976--2982},
}

@article{olfati-saber2004tac,
  author={R.~Olfati-Saber and R.~M.~Murray},
  journal={IEEE Trans. Autom. Control}, 
  title={Consensus problems in networks of agents with switching topology and time-delays}, 
  year={2004},
  volume={49},
  number={9},
  pages={1520--1533},
}

@article{panteley2017tac,
  author={E.~Panteley and A.~Lor\'ia},
  journal={IEEE Trans. Autom. Control}, 
  title={Synchronization and Dynamic Consensus of Heterogeneous Networked Systems}, 
  year={2017},
  volume={62},
  number={8},
  pages={3758--3773},
}

@article{lee2020aut,
  title = {A tool for analysis and synthesis of heterogeneous multi-agent systems under rank-deficient coupling},
  author = {J.~G.~Lee and H.~Shim},
  journal = {Automatica},
  volume = {117},
  note = {{A}rt.~no.~108952},
  year = {2020}
}

@article{munz2010aut,
  title = {Delay robustness in consensus problems},
  author = {U.~M\"unz and A.~Papachristodoulou and F.~Allg\"ower},
  journal = {Automatica},
  volume = {46},
  pages = {1252--1265},
  year = {2010}
}

@inproceedings{scardovi2008cdc,
    author = {L.~Scardovi and R.~Sepulchre},
    title = {Synchronization in networks of identical linear systems},
    booktitle = {Proc. IEEE Conf. Decision Control},
    year = {2008},
    pages = {546--551}
}

@inproceedings{moreau2004cdc,
    author = {L.~Moreau},
    title = {Stability of continuous-time distributed consensus algorithms},
    booktitle = {Proc. IEEE Conf. Decision Control},
    year = {2004},
    pages = {3998--4003}
}

@article{dimarogonas2012tac,
    author = {D.~V.~Dimarogonas and E.~Frazzoli and K.~H.~Johansson},
    title = {Distributed event-triggered control for multi-agent systems},
    journal = {IEEE Trans. Autom. Control},
    year = {2012},
    volume = {57},
    number = {5},
    pages = {1291--1297}
}

@article{yu2011tac,
    author = {W.~Yu and G.~Chen and M.~Cao},
    title = {Consensus in directed networks of agents with nonlinear dynamics},
    journal = {IEEE Trans. Autom. Control},
    year = {2011},
    volume = {56},
    number = {6},
    pages = {1436--1441}
}

@article{wieland2011aut,
  title = {An internal model principle is necessary and sufficient for linear output synchronization},
  author = {P.~Wieland and R.~Sepulchre and F.~Allg\"ower},
  journal = {Automatica},
  volume = {47},
  pages = {1068--1074},
  year = {2011}
}

@article{lohmiller1998aut,
    author = {W.~Lohmiller and J.-J.~E.~Slotine},
    title = {On contraction analysis for non-linear systems},
    journal = {Automatica},
    year = {1998},
    volume = {34},
    number = {6},
    pages = {683--696}
}

@article{lee2022aut,
    title = {Blended dynamics approach to distributed optimization: {S}um convexity and convergence rate},
    journal = {Automatica},
    volume = {141},
    year = {2022},
    author = {S.~Lee and H.~Shim},
    note = {{A}rt.~no.~110290},
}

@article{kim2016tac,
    title = {Robustness of synchronization of heterogeneous agents by strong coupling and a large number of agents},
    journal = {IEEE Trans. Autom. Control},
    volume = {61},
    year = {2016},
    author = {J.~Kim and J.~Yang and H.~Shim and J.-S.~Kim and J.~H.~Seo},
    number = {10},
    pages = {3096--3102}
}

@incollection{lee2022bookchapter,
author="J.~G.~Lee and H.~Shim",
editor="Jiang, Zhong-Ping
and Prieur, Christophe
and Astolfi, Alessandro",
title="Design of Heterogeneous Multi-agent System for Distributed Computation",
bookTitle="Trends in Nonlinear and Adaptive Control",
year="2022",
publisher="Springer International Publishing",
address="Cham, Switzerland",
pages="83--108",
}

@article{niebur1991PhysRevLett,
  title = {Collective frequencies and metastability in networks of limit-cycle oscillators with time delay},
  author = {E.~Niebur and H.~G.~Schuster and D.~M.~Kammen},
  journal = {Phys. Rev. Lett.},
  volume = {67},
  number = {20},
  pages = {2753--2756},
  year = {1991},
}

@article{torres2026PhysRevE,
  title = {Intrinsic resonance depends on network size for coupled-delayed interacting oscillators},
  author = {F.~A.~Torres and A.~Weinstein and J.~M.~Cortes and W.~El-Deredy},
  journal = {Phys. Rev. E},
  volume = {113},
  year = {2026},
  note = {{A}rt.~no.~024207},
}

@inproceedings{lee2018cdc,
  author={J.~G.~Lee and H.~Shim},
  title={Heterogeneous {V}an der {P}ol Oscillators under Strong Coupling},
  booktitle = {Proc. IEEE Conf. Decision Control},
  year={2018},
  pages={3666-3673},
}

\newpage
\appendices

\section{Proof of Lemmas}\label{section:proof_of_lemmas}
\subsection{Proof of Lemma~\ref{lem:property_of_E}}

A straightforward calculation shows the first claim. The second claim follows from
\begin{align*}
    \|\wt(\tau)\| &=\|(R^T R\otimes I_n)\wt(\tau)\|\\
    &\le  \sqrt{Nn}\|(R\otimes I_n)\wt(\tau)\|_\infty\\
    &= \sqrt{Nn}\|\xt(\tau)-\one_N\otimes \zt(\tau)\|_\infty \le \sqrt{Nn}E(\xt_\tau),
\end{align*}
where we used $\|R^T\otimes I_n\|=1$ and the relationship between the 2-norm and the infinity norm.

\subsection{Proof of Lemma~\ref{lem:evolution_bound}}

Fix $\ell$ and write $\xt_i = \mathrm{col}(\xt_i^1,\ldots,\xt_i^n), i=1,\ldots,N$. For each $\tau \in (\tau_0,\tau_1)$, let $\Jcal(\tau)$ be the set of every pair $(i,\tilde{\tau}) \in \{1,\ldots,N\}\times[\tau-\varepsilon d,\tau]$ such that $M^{\ell}(\xt_{\tau})=\xt_{i}^{\ell}(\tilde{\tau})$. If $(i,\tau) \notin \Jcal(\tau)$ for all $i=1,\ldots,N$, then we obtain $D^+M^{\ell}(\xt_{\tau}) \le 0$. If not, choose $i^* \in \{1,\ldots,N\}$ such that $(i^*,\tau) \in \Jcal(\tau)$ and $(\xt^{\ell}_{i^*})'(\tau) \ge (\xt^{\ell}_{i})'(\tau)$ for all $i$ satisfying $(i,\tau) \in \Jcal(\tau)$. If $(\xt^{\ell}_{i^*})'(\tau) \le 0$, then again we get $D^+M^{\ell}(\xt_{\tau}) \le 0$. Otherwise,
\begin{align*}
    &D^+M^{\ell}(\xt_{\tau}) = (\xt^{\ell}_{i^*})'(\tau) \\ 
    &\le B_0 +\frac{1}{\varepsilon}\sum_{j\in \Ncal_{i^*}} a_{i^*j}(\xt^{\ell}_j(\tau-\varepsilon d_{i^*j})-M^{\ell}(\xt_\tau)) \le B_0
\end{align*}
and the comparison principle proves the first inequality. The second inequality can be proved similarly.

\subsection{Proof of Lemma~\ref{lem:output_converse_Lyapunov_Krasovskii}}

Following the proof of \cite[Thm. 1]{munz2011tac} and using a time-rescaling argument, we can choose $M\ge1$ and $\lambda>0$, as functions of only the adjacency matrix and $d$, such that any solution $\vt(\tau)$ of \eqref{eq:integrator_tau} satisfies
\begin{align*}
    E(\vt_\tau) \le Me^{-\lambda \tau/\varepsilon}E(\vt_0),\qquad \forall \tau \ge 0.
\end{align*}
Now, motivated by the proof of \cite[Lemma 33.1]{krasovskii}, define a functional $V:C([-\varepsilon d,0],\R^{Nn}) \to [0,\infty)$ by
\begin{align}\label{eq:Lyapunov-Krasovskii_construction}
    V(\yt) \coloneqq \frac{1}{\varepsilon}\int_0^{\varepsilon T_0} E(\vt_s[\yt]) ds + \sup_{0 \le s \le \varepsilon T_0} E(\vt_s[\yt])
\end{align}
where $\vt[\yt]$ is the solution trajectory of \eqref{eq:integrator_tau} with $\vt[-\varepsilon d,0]=\yt$ and $T_0\coloneq \frac{1}{\lambda}\ln 2M$. The first property is verified by
\begin{align*}
    E(\yt) &\le \sup_{0 \le s \le \varepsilon T_0} E(\vt_s[\yt]) \le  V(\yt),\\
    V(\yt) &\le \frac{1}{\varepsilon}\int_0^{\varepsilon T_0} Me^{-\lambda s/\varepsilon}E(\yt)ds + ME(\yt) \\
    &= \bigg[ \frac{M}{\lambda}(1-e^{-\lambda T_0}) + M \bigg]E(\yt).
\end{align*}

Since we have $E(\vt_{\tau+\varepsilon T_0}) \le Me^{-\lambda T_0}E(\vt_\tau)=\frac{1}{2}E(\vt_\tau)$ and $\vt_s[\vt_\tau]=\vt_{s+\tau}$ for $s,\tau \ge 0$, the last term in \eqref{eq:Lyapunov-Krasovskii_construction} with $\yt=\vt_\tau$ does not increase with $\tau$. Hence, we can calculate
\begin{align*}
    &D^+V(\vt_\tau)\\
    &\le \limsup_{h \to 0+} \frac{1}{\varepsilon h}\bigg[ \int_{h}^{\varepsilon T_0+h} E(\vt_{s+\tau})ds- \int_0^{\varepsilon T_0} E(\vt_{s+\tau})ds\bigg]\\
    &=\frac{1}{\varepsilon}\left[ E(\vt_{\tau + \varepsilon T_0}) - E(\vt_\tau) \right] \le -\frac{1}{2\varepsilon} E(\vt_\tau),
\end{align*}
where we again used $E(\vt_{\tau+\varepsilon T_0}) \le \frac{1}{2}E(\vt_\tau)$. This proves the second property. The last property follows from
\begin{align*}
    &|V(\yt)-V(\tilde{\yt})|\\
    &\le \frac{1}{\varepsilon}\int_0^{\varepsilon T_0} E(\vt_s[\yt-\tilde{\yt}]) ds + \sup_{0 \le s \le \varepsilon T_0} E(\vt_s[\yt-\tilde{\yt}])\\
    & \le \bigg[ \frac{1}{\varepsilon}\int_0^{\varepsilon T_0} Me^{-\lambda s/\varepsilon} ds+ \sup_{0 \le s \le \varepsilon T_0} Me^{-\lambda s/\varepsilon} \bigg]E(\yt-\tilde{\yt})\\
    &\le \bigg[ \frac{M}{\lambda}(1-e^{-\lambda T_0}) + M \bigg]E(\yt-\tilde{\yt}),
\end{align*}
where we used that $E$ is a seminorm and \eqref{eq:integrator_tau} is linear.

\subsection{Proof of Lemma~\ref{lem:practical_consensus}}

Let $\yt \in C([-\varepsilon d,0],\R^{Nn})$. For $\xi \in \R^{Nn}$ and $0<h<\varepsilon d$, define a segment $\yt[h;\xi] \in C([-\varepsilon d,0],\R^{Nn})$ by
\begin{align*}
    \yt[h;\xi](\theta) \coloneqq
    \begin{cases}
        \yt(\theta+h) &\theta \in [-\varepsilon d,-h),\\
        \yt(0) + (\theta+h)\xi &\theta \in [-h,0].
    \end{cases}
\end{align*}
The following lemma is another version of \cite[Thm. 4.2.3]{burton}, which states that the Dini derivatives of a functional can be computed using first-order approximations.
\begin{lemma}[\protect{\cite[Thm. 4.2.3]{burton}}]\label{lem:Dini_trick}
    Let $V$ be as in Lemma~\ref{lem:output_converse_Lyapunov_Krasovskii}\footnote{As the proof shows, Lemma~\ref{lem:Dini_trick} holds for general locally Lipschitz functionals with respect to the supremum norm.} and let $\xt:\R \to \R^{Nn}$. If $\xt$ has right-hand derivative $\xi$ at $\tau$, i.e., $\lim_{h \to 0+} \frac{1}{h}\|\xt(\tau+h)-\xt(\tau)-h \xi\|_\infty=0$, then
    \begin{align*}
        D^+V(\xt_\tau)=\limsup_{h \to 0+} \frac{V(\xt_\tau[h;\xi])-V(\xt_\tau)}{h}.
    \end{align*}
\end{lemma}
\begin{proof}
    By the definition of $D^+V(\xt_{\tau})$, it suffices to show that $\lim_{h \to 0+} \frac{1}{h}|V(\xt_{\tau+h})-V(\xt_\tau[h;\xi])| =0$.
    Note that the values of $\xt_{\tau+h}(\theta)$ and $\xt_\tau[h;\xi](\theta)$ differ only for $\theta \in (-h,0]$. Therefore, we obtain
    \begin{align*}
            & \frac{1}{h}|V(\xt_{\tau+h})-V(\xt_\tau[h;\xi])| \le \frac{c_4}{h}E(\xt_{\tau+h}-\xt_\tau[h;\xi]) \\
            &\le \frac{2c_4}{h}\sup_{\theta \in (-h,0]}\|\xt_{\tau+h}(\theta)-\xt_\tau[h;\xi](\theta)\|_{\infty} \\
            &= \frac{2c_4}{h}\sup_{\tilde{\theta} \in (0,h]} \|\xt(\tau+\tilde{\theta})-\xt(\tau)-\tilde{\theta} \xi\|_{\infty}\\
            &\le \sup_{\tilde{\theta} \in (0,h]} \frac{2c_4}{\tilde{\theta}}\|\xt(\tau+\tilde{\theta})-\xt(\tau)-\tilde{\theta} \xi\|_{\infty} \xrightarrow{h \to 0+} 0
    \end{align*}
    where we used Lemma~\ref{lem:property_of_E} and the change of variables $\tilde{\theta}=\theta+h$.
    For the last part, we applied the definition of the right-hand derivative.
\end{proof}

Define $G_i \coloneqq \frac{1}{\varepsilon}\sum_{j \in \Ncal_i}a_{ij}(\xt_j(\tau-\varepsilon d_{ij}) - \xt_i(\tau))$, $G \coloneqq \mathrm{col}(G_1,\ldots, G_N)$, and $G^* \coloneqq f(\tau,\xt(\tau)) + G$. Note that $\xt$ has right-hand derivative $G^*$ at $\tau$, whereas $G$ is the vector field of \eqref{eq:integrator_tau} with $\vt$ replaced by $\xt$. Then by Lemma~\ref{lem:Dini_trick},
\begin{align*}
    &D^+V(\xt_\tau)= \limsup_{h \to 0+} \frac{V(\xt_{\tau}[h;G^*])-V(\xt_\tau)}{h}\\
    &\le \limsup_{h \to 0+}\frac{V(\xt_{\tau}[h;G])-V(\xt_\tau)}{h}\\
    &\qquad+\limsup_{h \to 0+}\frac{V(\xt_{\tau}[h;G^*])-V(\xt_{\tau}[h;G])}{h}\\
    &\le -\frac{1}{2\varepsilon}E(\xt_\tau) + \limsup_{h \to 0+} \frac{c_4}{h} E\left( \xt_{\tau}[h;G^*]-\xt_{\tau}[h;G] \right)\\
    &\le -\frac{1}{2\varepsilon}E(\xt_\tau) + \limsup_{h \to 0+} \sup_{\tilde{\theta} \in (0,h]} \frac{2c_4}{h}\| \tilde{\theta} f(\tau,\xt(\tau))\|_{\infty}\\
    &\le  -\frac{1}{2\varepsilon c_2}V(\xt_\tau) + 2c_4\|f(\tau,\xt(\tau))\|_\infty.
\end{align*}
This proves the first claim. The next claim follows directly from the comparison principle and properties of $V$.

\section{Complete Proof of Theorem~\ref{thm:main_result}}\label{section:complete_proof_of_theorem_1}

\subsection{Preparation}\label{subsection:preparation}

By Assumption~\ref{ass:contractivity}, any two solutions of $\eqref{eq:blended_dynamics}$ approach each other \cite{lohmiller1998aut}. Therefore, it suffices to fix $\st(0) \in S$ where $S$ is defined in Lemma~\ref{lem:s_dynamics}. By Lemma~\ref{lem:s_dynamics}, $\st(\tau) \in S$ for all $\tau \ge 0$.
Choose a compact rectangle $K_1 \subseteq \R^n$ such that $K \cup S \subsetneq \operatorname{int} K_1$. Let $r_1>\mathrm{diam}(K_1) \coloneq \sup_{p,q \in K_1}\|p-q\|$. Then for any $d_{ij} \in [0,d]$, it follows that
\begin{align}\label{eq:K_1_and_r_1}
    \|p-\st(\tau)\| \le r_1, \qquad \forall p \in  K_1,\forall \tau \ge 0.
\end{align}
Next, define a compact, convex set $K_2 \subseteq \R^n$ by
\begin{align}\label{eq:K_2_and_r_2}
    K_2 \coloneqq \bigcup_{q\in K_1} \left\{ p \in \R^n : \|p-q\| \le r_2 +1 \right\},
\end{align}
where $r_2 \coloneq \sqrt{\frac{\|H\|}{\lambda_{\min}(H)}}(1+ 2\bar{a}d)r_1$, $\bar{a}=\frac{1}{N}\sum_{i,j=1}^N a_{ij}$, and $H$ is defined in Assumption~\ref{ass:contractivity}. The specific choices of $r_1$, $r_2$, $K_1$, and $K_2$ are needed only in the final part of the proof. Until then, it suffices to note that $K \subsetneq \operatorname{int} K_1$, $ K_1 \subsetneq \operatorname{int} K_2$, and these sets do not depend on $\varepsilon$ or the values of $d_{ij}$.

Let $B_0 \coloneq \sup\{\|f_i(\tau,\xt_i)\|_\infty:\xt_i \in K_2, i=1,\ldots,N,\tau \ge 0\}$, which is finite by Assumption~\ref{ass:vector_field}.
Since $K \subsetneq \operatorname{int} K_1$, Lemma~\ref{lem:evolution_bound} yields a constant $T>0$ such that, for any choice of $d_{ij} \in [0,d]$ and $\varepsilon  \in (0,1]$, one has $\xt_i(\tau) \in K_1$ for all $i$ and $\tau \in [0,T]$.
For later use, we split the interval $[0,T]$ by setting $T_1\coloneqq T/3$ and $T_2 \coloneqq 2T/3$.

Although we are ready to choose the desired $\varepsilon^*$, we defer its specification to the following subsections for clarity.
The reader may assume that $\varepsilon^*$ has been chosen here, and later verify that such a choice can indeed be made, independently of the values of $d_{ij}$. Henceforth, fix any collection of delays $d_{ij}\in [0,d]$. For $\varepsilon>0$, let $\overline{T}_{\varepsilon}$ be the first time at which any $\xt_i(\tau),i=1,\ldots,N$ leaves $K_2$. By construction, $\overline{T}_{\varepsilon}>T$ and $\xt_i(\tau) \in K_2$ for all $i$ and $\tau \in [0,\overline{T}_\varepsilon)$. In particular, $\|f_i(\tau,\xt_i(\tau))\|_\infty \le B_0$ for all $i$ and $\tau \in [0,\overline{T}_\varepsilon)$.
We later show that $\overline{T}_{\varepsilon}=\infty$ whenever $0<\varepsilon \le \varepsilon^*$ by a contradiction argument.

\subsection{Practical Synchronization}

Let $V$ be the functional defined in Lemma~\ref{lem:output_converse_Lyapunov_Krasovskii}. Using $V \le c_2 E$, choose an upper bound $V_0$ for $V(\xt[-\varepsilon d,0])$ over all initial histories $\xt_i[-\varepsilon d,0] \in C([-\varepsilon d,0],K)$, $i=1,\ldots,N$. Note that $V_0$ can be chosen finite, independently of $\varepsilon$. For $\tau \in [0,\overline{T}_{\varepsilon})$, Lemma~\ref{lem:practical_consensus} gives
\begin{align*}
    E(\xt_\tau) \le \exp\left(-\frac{\tau}{2\varepsilon c_2} \right)V_0 + 4\varepsilon c_2 c_4 B_0.
\end{align*}
When $\tau=T_1$, we obtain
\begin{align*}
    \lim_{\varepsilon \to 0+} \frac{1}{\varepsilon}\exp\left(-\frac{T_1}{2\varepsilon c_2} \right)V_0=0.
\end{align*}
Therefore, $E(\xt_\tau)=O(\varepsilon)$ as $\varepsilon \to 0+$, uniformly for $\tau \in [T_1,\overline{T}_{\varepsilon})$. Choose $\varepsilon^*_1>0$ such that $0<\varepsilon \le \varepsilon^*_1$ implies
\begin{align}\label{eq:consensus_error_x}
\begin{aligned}
    E(\xt_\tau) &\le 8\varepsilon c_2 c_4 B_0 \eqqcolon \bar{B}_0 \varepsilon,\\
    \|\wt(\tau)\| &\le \sqrt{Nn}E(\xt_\tau) \le 1,\\
    \varepsilon d & \le T/3=T_1=T_2-T_1,
\end{aligned}
\qquad\forall \tau \in [T_1,\overline{T}_{\varepsilon}).
\end{align}
Since the constants $c_2,c_4,T_1,V_0,$ and $B_0$ do not depend on the values of $d_{ij}$, $\varepsilon^*_1$ can be chosen independently of $d_{ij}$. From now on, we assume $0<\varepsilon \le \varepsilon^*_1$ so that \eqref{eq:consensus_error_x} holds.

For each $h \in [0,d]$, define the difference quotient
\begin{align*}
    \qt^{h}_i(\tau) \coloneqq \frac{\xt_i(\tau-\varepsilon h)-\xt_i(\tau)}{\varepsilon}
\end{align*}
and set $\qt^{h} \coloneqq \mathrm{col}(\qt_1^{h},\ldots,\qt_N^{h})$. We show that $\{\qt^h_i\}_{i=1}^N$ also achieves practical synchronization. Note that \eqref{eq:consensus_error_x} implies
\begin{align}\label{eq:q_bound}
\begin{split}
    \|\qt^{h}_i(\tau)\|_\infty &=\max_{\ell = 1,\ldots,n}\left| \frac{\xt_i^{\ell}(\tau-\varepsilon h)-\xt_i^{\ell}(\tau)}{\varepsilon} \right| \\
    &\le \frac{1}{\varepsilon}E(\xt_\tau) \le  \bar{B}_0, \qquad \forall\tau \in [T_1,\overline{T}_{\varepsilon}).
\end{split}
\end{align}
By the linearity of the diffusive coupling, $\{\qt_i^{h}\}_{i=1}^N$ inherits the same coupling structure; since $\varepsilon d \le T_1$ by \eqref{eq:consensus_error_x}, for $\tau \in  [T_1,\overline{T}_\varepsilon)$ we have
\begin{align}\label{eq:q_dynamics}
\begin{aligned}
    (\qt_i^{h})'(\tau)
    &=\frac{1}{\varepsilon}[f_i(\tau-\varepsilon h,\xt_i(\tau-\varepsilon h))-f_i(\tau,\xt_i(\tau))]\\
    &\quad +\frac{1}{\varepsilon}\sum_{j \in \Ncal_i}a_{ij}(\qt_j^{h}(\tau-\varepsilon d_{ij})-\qt_i^{h}(\tau)).
\end{aligned}
\end{align}

We bound the first term on the right-hand side of \eqref{eq:q_dynamics}. By Assumptions~\ref{ass:vector_field} and \ref{ass:Lipschitz_continuity}, there exists a Lipschitz constant $L_{K_2}>0$ (with respect to $\|\cdot\|_\infty$) such that each $f_i$ is Lipschitz in $t \in [0,\infty)$ uniformly in $x_i \in K_2$, and Lipschitz in $x_i \in K_2$ uniformly in $t$. Hence, using \eqref{eq:q_bound}, we obtain
\begin{align*}
    &\frac{1}{\varepsilon}\left\| f_i(\tau-\varepsilon h,\xt_i(\tau-\varepsilon h))-f_i(\tau,\xt_i(\tau)) \right\|_\infty\\
    &\le \frac{1}{\varepsilon}\|f_i(\tau-\varepsilon h,\xt_i(\tau-\varepsilon h))-f_i(\tau,\xt_i(\tau-\varepsilon h))\|_\infty\\
    &\quad+\frac{1}{\varepsilon}\|f_i(\tau,\xt_i(\tau-\varepsilon h))-f_i(\tau,\xt_i(\tau))\|_\infty\\
    &\le L_{K_2} (h + \|\qt_i^h(\tau)\|_\infty)\\
    &\le L_{K_2} (d + \bar{B}_0) \eqqcolon B_1, \qquad \forall\tau \in [T_1,\overline{T}_{\varepsilon}).
\end{align*}
 Next, we use the functional $V$ to estimate the synchronization error for $\qt^h$. Applying Lemma~\ref{lem:practical_consensus} on \eqref{eq:q_dynamics} instead of \eqref{eq:tau_scale}, we obtain, for $\tau \in [T_2,\overline{T}_{\varepsilon})$,
\begin{align*}
    E(\qt^h_\tau) &\le \exp\left(-\frac{\tau - T_2}{2\varepsilon c_2} \right)V(\qt^h_{T_2}) + 4\varepsilon c_2 c_4B_1.
\end{align*}
Using \eqref{eq:q_bound} and $\varepsilon d \le T_2-T_1$ from \eqref{eq:consensus_error_x}, we have
\begin{align*}
    V(\qt^h_{T_2}) &\le c_2E(\qt^h_{T_2})\\
    &\le 2c_2\sup_{\theta \in [-\varepsilon d,0]}\|\qt^h(T_2+\theta)\|_\infty \le 2c_2\bar{B}_0.
\end{align*}
By the same argument as before, we conclude that $E(\qt^h_\tau)=O(\varepsilon)$ as $\varepsilon \to 0+$, uniformly in $\tau \in [T,\overline{T}_{\varepsilon})$ and $h \in [0,d]$. Choose $\varepsilon^*_2>0$ such that $0 < \varepsilon \le \varepsilon^*_2$ implies 
\begin{align}\label{eq:consensus_error_q}
    E(\qt^h_\tau) &\le  8\varepsilon c_2 c_4 B_1 \eqqcolon \bar{B}_1 \varepsilon, \qquad \forall\tau \in [T,\overline{T}_{\varepsilon})
\end{align}
uniformly in $h \in [0,d]$. Note that $\varepsilon^*_2$ can be chosen independently of the values of $d_{ij}$, similar to the case of $\varepsilon^*_1$.

\subsection{Approximation by Blended Dynamics}\label{subsection:approximation_by_blended_dynamics}

We now estimate the difference between $\zt$ and $\st$. We omit the argument $\tau$ (e.g., $\zt$ instead of $\zt(\tau)$) when it is clear from the context; delayed terms are always written explicitly. Assume $0<\varepsilon \le \min\{\varepsilon^*_1,\varepsilon^*_2\}$ so that \eqref{eq:consensus_error_x} and \eqref{eq:consensus_error_q} hold. Let $F(\tau,\zt,\wt)\coloneq \frac{1}{N}\sum_{i=1}^N f_i(\tau,\zt+(R_i\otimes I_n)\wt)$. Using \eqref{eq:blended_dynamics} and \eqref{eq:z_dynamics}, and adding and subtracting $F(\tau,\zt,0)$ and $\frac{1}{N}\sum_{i,j=1}^N a_{ij} \frac{\st(\tau-\varepsilon d_{ij})-\st}{\varepsilon}$, we obtain
\begin{align*}
    &\zt'-\st'=\left[ F(\tau,\zt,0) -F(\tau,\st,0) \right]+[F(\tau,\zt,\wt)-F(\tau,\zt,0)]\\
    &\quad+\frac{1}{\varepsilon N}\sum_{i,j=1}^N a_{ij}\left[\zt(\tau-\varepsilon d_{ij})-\st(\tau-\varepsilon d_{ij})-(\zt-\st)\right]\\
    &\quad+\frac{1}{N}\sum_{\substack{i,j=1 \\ d_{ij} \neq 0}}^N a_{ij}d_{ij}\left[\frac{1}{1+A}F(\tau,\st,0)- \frac{\st(\tau-\varepsilon d_{ij})-\st}{-\varepsilon d_{ij}} \right]\\
    &\quad +\frac{1}{N} \sum_{i,j=1}^N a_{ij}(R_j \otimes I_n)\frac{\wt(\tau-\varepsilon d_{ij})-\wt}{\varepsilon}.
\end{align*}
Note that Assumption~\ref{ass:contractivity} gives
\begin{align*}
    (\zt-\st)^TH\left[ F(\tau,\zt,0) -F(\tau,\st,0) \right] \le -\frac{\gamma}{2}\|\zt-\st\|^2_H.
\end{align*}
Also, when $d_{ij} \neq 0$ and $\tau \ge T$, \eqref{eq:blended_dynamics} gives the relation
\begin{align*}
    \frac{\st(\tau-\varepsilon d_{ij})-\st}{-\varepsilon d_{ij}}=\frac{1}{\varepsilon d_{ij}} \int_{\tau-\varepsilon d_{ij}}^\tau \frac{1}{1+A}F(\theta,\st(\theta),0)d \theta,
\end{align*}
and this leads to
\begin{align*}
    &\bigg\|\frac{1}{1+A}F(\tau,\st(\tau),0)- \frac{\st(\tau-\varepsilon d_{ij})-\st(\tau)}{-\varepsilon d_{ij}} \bigg\|_H\\
    &=\frac{1}{\varepsilon d_{ij}} \bigg\| \frac{1}{1+A} \int_{\tau -\varepsilon d_{ij}}^{\tau} (F(\tau,\st(\tau),0)-F(\theta,\st(\theta),0)) d\theta\bigg\|_H\\
    &\le \frac{1}{1+A}\sup_{\theta \in [\tau-\varepsilon d,\tau]}\|F(\tau,\st(\tau),0)-F(\theta,\st(\theta),0)\|_H.
\end{align*}

Let $W_1(\tau)\coloneq \|\zt(\tau)-\st(\tau)\|^2_H$. By the above results and the Cauchy-Schwarz inequality, we obtain for $\tau \in [T,\overline{T}_\varepsilon)$,
\begin{align*}
    &W'_1=2(\zt-\st)^TH(\zt'-\st')\\
    &\le -\gamma W_1 + 2\sqrt{W_1}\sqrt{\|H\|}\left\| \frac{\partial F}{\partial \wt}(\tau,\zt,r\wt) \right\|\|\wt\|\\
    &+ \frac{2\sqrt{W_1}}{\varepsilon N}\sum_{i,j=1}^N a_{ij}\left( \sqrt{W_1(\tau-\varepsilon d_{ij})}-\sqrt{W_1} \right)\\
    &+2\sqrt{W_1} \frac{\bar{a}d}{1+A} \sup_{\theta \in [\tau-\varepsilon d,\tau]}\|F(\tau,\st,0)-F(\theta,\st(\theta),0)\|_H\\
    &+\frac{2\sqrt{W_1}}{N}\sqrt{\|H\|}\sum_{i,j=1}^N a_{ij}\left\| \frac{\wt(\tau- \varepsilon d_{ij})-\wt}{\varepsilon}\right\|,
\end{align*}
where $r \in (0,1)$ is chosen by the mean value theorem.

To cancel out the delayed term, define
\begin{align*}
    W_2(\tau) \coloneqq \frac{1}{\varepsilon N} \sum_{i,j=1}^N a_{ij} \int_{\tau-\varepsilon d_{ij}}^\tau e^{-c(\tau-\varepsilon d_{ij}-\theta)}\sqrt{W_1(\theta)}\ d\theta,
\end{align*}
where $c\coloneqq \gamma/(2+3\bar{a}d)>0$, $\gamma$ is in Assumption~\ref{ass:contractivity}, and $\bar{a}=\frac{1}{N}\sum_{i,j=1}^N a_{ij}$. Set $W\coloneq\sqrt{W_1} +W_2$. Let $l>0$ be an upper bound for $\left\| \frac{\partial F}{\partial \wt}(\tau,\zt,\wt) \right\|$ over all $\tau \ge 0$, $\zt \in K_2$, and $\|\wt\| \le 1$.
When $0<\varepsilon \le \min\{\varepsilon^*_1,\varepsilon^*_2\}$ and $\tau \in [T,\overline{T}_{\varepsilon})$, we have $\|\wt\| \le 1$ by \eqref{eq:consensus_error_x} and $\zt \in K_2$ by the definition of $\overline{T}_\varepsilon$. Therefore, in this case, we obtain
\begin{align}\label{eq:DW}
\begin{split}
    &D^+W = D^+\sqrt{W_1} -cW_2\\
    &\qquad\quad+ \frac{1}{\varepsilon N}\sum_{i,j=1}^N a_{ij}\bigg(e^{\varepsilon c d_{ij}}\sqrt{W_1}-\sqrt{W_1(\tau-\varepsilon d_{ij})}\bigg)\\
    &\le \bigg( -\frac{\gamma}{2} + \frac{1}{N}\sum_{i,j=1}^N a_{ij}\frac{e^{\varepsilon c d_{ij}}-1}{\varepsilon} \bigg)\sqrt{W_1} - cW_2\\
    &\quad+\frac{\bar{a}d}{1+A} \sup_{\theta \in [\tau-\varepsilon d,\tau]}\|F(\tau,\st,0)-F(\theta,\st(\theta),0)\|_H\\
    &\quad+ l\sqrt{\|H\|}\,\|\wt\|+\frac{\sqrt{\|H\|}}{N}\sum_{i,j=1}^N a_{ij}\left\| \frac{\wt(\tau- \varepsilon d_{ij})-\wt}{\varepsilon}\right\|\\
    &\le \bigg(-\frac{\gamma}{2} + \bar{a}\frac{e^{\varepsilon c d}-1}{\varepsilon} \bigg)\sqrt{W_1} - cW_2\\
    &\quad+\bar{a}d\sup_{\theta \in [\tau-\varepsilon d,\tau]}\|F(\tau,\st,0)-F(\theta,\st(\theta),0)\|_H\\
    &\quad+ l\sqrt{\|H\|}\,\|\wt\|+\bar{a}\sqrt{\|H\|}\max_{i,j} \left\| \frac{\wt(\tau- \varepsilon d_{ij})-\wt}{\varepsilon}\right\|.
\end{split}
\end{align}

We focus on the last expression in \eqref{eq:DW}. We show that for sufficiently small $\varepsilon$, the coefficient of $\sqrt{W_1}$ is negative, while the last three terms are $O(\varepsilon)$ as $\varepsilon \to 0+$ uniformly for $\tau \in [T,\overline{T}_{\varepsilon})$. Note that $\lim_{\varepsilon \to 0+} (e^{\varepsilon c d}-1)/\varepsilon = cd$.
Choose $\varepsilon^*_3>0$, independently of $d_{ij}$, such that $0<\varepsilon \le \varepsilon^*_3$ implies $(e^{\varepsilon c d}-1)/{\varepsilon}<3cd/{2}$.
Set
\begin{align}\label{eq:epsilon_star}
    \varepsilon^*\coloneq \min \left\{\varepsilon^*_1,\varepsilon^*_2,\varepsilon^*_3,\frac{\ln2}{cd}, \frac{1}{2\beta_0}, 1\right\} >0,
\end{align}
where $\beta_0>0$ will be specified later; see \eqref{eq:C_1} and \eqref{eq:last_bound} for the explicit form of $\beta_0$. From now on, assume $0<\varepsilon \le \varepsilon^*$.
Since $-\frac{\gamma}{2}+\frac{3}{2}c\bar{a}d=-c$ from the choice of $c$, the condition $0<\varepsilon \le  \varepsilon^*$ simplifies \eqref{eq:DW} to
\begin{align}\label{eq:DW_simplified}
\begin{split}
    D^+W 
    & \le -cW + l\sqrt{\|H\|}\,\|\wt\|\\
    & \quad+\bar{a}d\sup_{\theta \in [\tau-\varepsilon d,\tau]}\|F(\tau,\st,0)-F(\theta,\st(\theta),0)\|_H\\
    &\quad+\bar{a}\sqrt{\|H\|}\max_{i,j}\left\| \frac{\wt(\tau- \varepsilon d_{ij})-\wt}{\varepsilon}\right\|.
\end{split}
\end{align}

By Assumptions~\ref{ass:vector_field} and \ref{ass:Lipschitz_continuity}, $F(\tau,\st,0)$ is Lipschitz in $\tau \in [0,\infty)$ uniformly in $\st \in S$ and Lipschitz in $\st \in S$ uniformly in $\tau$. Let $L_S>0$ be a common Lipschitz constant of $F(\tau,\st,0)$ for both cases (with respect to $\|\cdot\|$). Also, by Lemma~\ref{lem:s_dynamics}, we can choose a Lipschitz constant $L_{\st}>0$ of the solution $\st(\tau)$ on $[0,\infty)$, independently of $d_{ij}$. Using these facts, we obtain
\begin{align}\label{eq:differentiation_approximation}
\begin{split}
    &\sup_{\theta \in [\tau-\varepsilon d,\tau]}\|F(\tau,\st(\tau),0)-F(\theta,\st(\theta),0)\|_H \\
    & \quad \le \sup_{\theta \in [\tau-\varepsilon d,\tau]}\|F(\tau,\st(\tau),0)-F(\theta,\st(\tau),0)\|_H \\
    &\quad\quad +\sup_{\theta \in [\tau-\varepsilon d,\tau]}\|F(\theta,\st(\tau),0)-F(\theta,\st(\theta),0)\|_H \\
    & \quad \le \sqrt{\|H\|} L_S(1+L_{\st})d\varepsilon, \qquad  \forall \tau \in  [T,\overline{T}_\varepsilon).
\end{split}
\end{align}
Moreover, using \eqref{eq:consensus_error_x}, \eqref{eq:consensus_error_q} and Lemma~\ref{lem:property_of_E}, $0<\varepsilon\le \varepsilon^*$ implies
\begin{align}\label{eq:synchronization_condition}
\begin{split}
    &\|\wt(\tau)\|\le \sqrt{Nn}E(\xt_\tau) \le \sqrt{Nn} \bar{B}_0 \varepsilon,\\
    &\max_{i,j}\left\| \frac{\wt(\tau- \varepsilon d_{ij})-\wt(\tau)}{\varepsilon}\right\| = \max_{i,j}\|(R^T\otimes I_n)\qt^{d_{ij}}(\tau)\|\\
    &\le \max_{i,j}\sqrt{Nn}E(\qt_{\tau}^{d_{ij}})\le \sqrt{Nn} \bar{B}_1 \varepsilon, \qquad \forall \tau \in [T,\overline{T}_{\varepsilon}).
\end{split}
\end{align}
Then, by substituting \eqref{eq:differentiation_approximation} and \eqref{eq:synchronization_condition} into \eqref{eq:DW_simplified}, we obtain
\begin{align*}
    D^+W(\tau) &\le -cW(\tau) +  \beta_1 \varepsilon        ,\qquad \forall \tau \in [T,\overline{T}_{\varepsilon})
\end{align*}
where we defined
\begin{align}\label{eq:C_1}
    \beta_1 \coloneqq  \sqrt{\|H\|} \left[ \sqrt{Nn} (l \bar{B}_0 + \bar{a}\bar{B}_1 )+  \bar{a}d^2 L_S(1+L_{\st}) \right].
\end{align}
By the comparison principle,
\begin{align}\label{eq:bound_for_W}
    W(\tau) \le e^{-c(\tau-T)}W(T)  + \frac{\beta_1}{c}\varepsilon, \qquad \forall \tau \in [T,\overline{T}_{\varepsilon}).
\end{align}

We now estimate $W(T)$. By the choice of $T$, $\zt(\tau) \in K_1$ for $\tau \in [0,T]$. Using \eqref{eq:K_1_and_r_1}, we have $\sqrt{W_1(\tau)} =\|\zt(\tau)-\st(\tau)\|_H \le r_1\sqrt{\|H\|}$ for $\tau \in [0,T]$. Moreover, this yields
\begin{align*}
    W_2(T) &\le \frac{1}{\varepsilon N} \sum_{i,j=1}^N a_{ij} \int_{T-\varepsilon d_{ij}}^T e^{-c(T-\varepsilon d_{ij}-\theta)}r_1\sqrt{\|H\|}\ d\theta\\
    &\le\frac{1}{N}\sum_{i,j=1}^N a_{ij} d_{ij}e^{\varepsilon c d}r_1\sqrt{\|H\|} \le 2\bar{a}dr_1\sqrt{\|H\|},
\end{align*}
where we used $\varepsilon\le \ln2/(cd)$. Hence, it follows that
\begin{align}\label{eq:WT_bound}
    W(T)=\sqrt{W_1(T)} + W_2(T) \le \sqrt{\|H\|}(1+2\bar{a}d)r_1.
\end{align}

Finally, using \eqref{eq:synchronization_condition}, \eqref{eq:bound_for_W}, and \eqref{eq:WT_bound}, we obtain
\begin{align}\label{eq:last_bound}
\begin{split}
    &\|\xt_i(\tau)-\st(\tau)\| \le \|\zt(\tau)-\st(\tau)\| + \|(R_i\otimes I_n)\wt(\tau)\|\\
    &\le \frac{W(\tau)}{\sqrt{\lambda_{\min}(H)}} + \|\wt(\tau)\|\\
    &\le e^{-c(\tau-T)}\frac{W(T)}{\sqrt{\lambda_{\min}(H)}} + \bigg( \underbrace{ \frac{\beta_1}{c\sqrt{\lambda_{\min}(H)}} + \sqrt{Nn} \bar{B}_0 }_{\eqqcolon \beta_0}\bigg) \varepsilon\\
    &\le e^{-c(\tau-T)}\sqrt{\frac{\|H\|}{\lambda_{\min}(H)}}(1+2\bar{a}d)r_1 + \beta_0 \varepsilon \\
    &\le e^{-c(\tau-T)}r_2 + \beta_0 \varepsilon \\
    &\le r_2 + \frac{1}{2}, \qquad \forall \tau \in [T,\overline{T}_{\varepsilon}),
\end{split}
\end{align}
where we used $\varepsilon \le 1/(2\beta_0)$ for the last inequality. Note that $\beta_0$ is determined independently of $\varepsilon^*$, which allows the choice \eqref{eq:epsilon_star} to be valid. Recall from Appendix~\ref{subsection:preparation} that $\st(\tau) \in K_1$ for all $\tau \ge 0$.
If $\overline{T}_{\varepsilon} < \infty$, then each $\xt_i(\overline{T}_{\varepsilon})$ lies in the interior of $K_2$ by \eqref{eq:K_2_and_r_2} and \eqref{eq:last_bound}, which contradicts the definition of $\overline{T}_\varepsilon$. Therefore, $\overline{T}_{\varepsilon}=\infty$ whenever $0 <\varepsilon \le \varepsilon^*$.
Let $\tau \to \infty$ in the second-to-last inequality of \eqref{eq:last_bound} to complete the proof.


\end{document}